\documentclass[12pt]{iopart}

\usepackage{graphicx}
\usepackage{rotating}

\usepackage{iopams}

\newtheorem{theorem}{Theorem}

\newtheorem{lemma}{Lemma}

\newtheorem{proposition}{Proposition}

\newenvironment{proof}[1][Proof]{\textbf{#1.} }{\ \rule{0.5em}{0.5em}}

\begin{document}

\title[]{Virtual copies of semisimple Lie algebras in enveloping algebras of semidirect products
and Casimir operators}

\author{R. Campoamor-Stursberg\dag}
\address{\dag\ Dpto. Geometr\'{\i}a y Topolog\'{\i}a\\Fac. CC. Matem\'aticas\\
Universidad Complutense de Madrid\\Plaza de Ciencias, 3\\E-28040
Madrid, Spain} \ead{rutwig@mat.ucm.es}

\author{S. G. Low\ddag}
\address{\ddag Austin, TX, USA}
\ead{Stephen.Low@alumni.utexas.net}

\begin{abstract}
Given a semidirect product $\frak{g}=\frak{s}\uplus\frak{r}$ of
semisimple Lie algebras $\frak{s}$ and solvable algebras
$\frak{r}$, we construct polynomial operators in the enveloping
algebra $\mathcal{U}(\frak{g})$ of $\frak{g}$ that commute with
$\frak{r}$ and transform like the generators of $\frak{s}$, up to
a functional factor that turns out to be a Casimir operator of
$\frak{r}$. Such operators are said to generate a virtual copy of
$\frak{s}$ in $\mathcal{U}(\frak{g})$, and allow to compute the
Casimir operators of $\frak{g}$ in closed form, using the
classical formulae for the invariants of $\frak{s}$. The behavior
of virtual copies with respect to contractions of Lie algebras is
analyzed. Applications to the class of Hamilton algebras and their
inhomogeneous extensions are given.

\end{abstract}

\pacs{02.20Sv, 02.20Qs}

\maketitle


\newpage

\section{Introduction}

Casimir operators of Lie algebras and their generalizations
constitute an important tool in dealing with various branches of
theoretical physics. Typically their appear in connection with
quantum numbers depicting the main properties of a system. Since
the eigenvalues of such operators characterize irreducible
representations and constitute measurable physical quantities,
they provide natural labels for classification schemes and
characterization of states. Further, connections between
symmetries and dynamical properties of a system lead to the
possibility of expressing Hamiltonians and deducing mass formulae
from the Casimir operators of the involved Lie algebras. This fact
has been extensively used in the algebraic models in atomic and
nuclear physics \cite{IA}. Applications of invariants of Lie
algebras to phenomenological aspects of a theory have also been
shown to be useful, as given for example by the cubic
$SU_{c}(3)$-operator in the confinement problem of the
non-relativistic quark model or the construction of the relevant
operators in $N=2$ supersymmetric models by means of Casimir
operators \cite{Wit}.

\medskip

Classical Casimir operators are operators from the universal
enveloping algebra $\mathcal{U}(\frak{g})$ that commute with all
elements of the Lie algebra, and therefore lead to the concept of
symmetry. For semisimple Lie algebras $\frak{s}$, the problem of
finding the centre of $\mathcal{U}(\frak{s})$ has been solved long
ago \cite{Ra,Per,Mac}, using the important structural properties
that are essential for their classification \cite{Co}. Although
polynomial invariants of an arbitrary Lie algebra always belong to
the centre of $\mathcal{U}(\frak{g})$, for non-semisimple algebras
this direct approach is quite complicated in practice, due to the
absence of structural properties like the Killing form and the
difficulties of their representation theory. For these types of
algebras, an analytic approach has shown to be more effective
\cite{Be,AA,Sh,C23,C33,Bo,C73}.

\medskip

However, the construction of Casimir operators by means of
enveloping algebras can be adapted well to semidirect products
$\frak{ws}=\frak{s}\overrightarrow
{\frak{\oplus}}_{R}\frak{w}\left( n\right)$ of simple Lie algebras
$\frak{s}$ with the Heisenberg-Weyl algebra $\frak{w}\left(
n\right)$. The main idea is to find quadratic operators in the
enveloping algebra that transform like the generators of the Levi
subalgebra $\frak{s}$ and commute with the generators of
$\frak{w}\left( n\right)$. In this manner, the Casimir operators
can be easily recovered using the classical formulae for the
invariants of $\frak{s}$. This compatibility is partially
explained by the simplicity of the boson realizations for these
algebras, and generally fails for radicals other than the
Heisenberg algebras.

\medskip

In this work, we propose a natural generalization of this method
to determine the Casimir operators of wide classes of semidirect
products of semisimple and solvable Lie algebras. The main
difference with the $\frak{ws}$ case is that the operators in the
enveloping algebra are of order higher than quadratic. These
operators commute with the generator of the radical $\frak{r}$,
and transform like the generators of the Levi subalgebra
$\frak{s}$, up to a functional factor given by a Casimir operator
of $\frak{g}$ depending only on the generators of the radical. It
follows that the normalized operators do not live properly in the
enveloping algebra of $\frak{g}$, but in its fraction field
\cite{AA}. For this reason such copies will be called virtual.
This modification does however not affect the computation of the
Casimir operators using the formulae for the invariants of
$\frak{s}$.

\medskip

As an application of this method, we obtain the Casimir operators
for the inhomogeneous Hamilton algebras $IHa(N)$ and all its
central extensions in arbitrary dimension. These algebras have
been used in the group theoretical analysis of noninertial states
in Quantum Mechanics \cite{SL3}. Some unanswered questions
formulated in \cite{SL1} concerning the number and form of their
invariants are solved simultaneously for all values of $N$.

\medskip

Finally, given a contraction
$\frak{g}\rightsquigarrow\frak{g}^{\prime}$, we analyze under
which conditions the operators spanning a virtual copy of the Levi
subalgebra $\frak{s}\subset\frak{g}$ in $\mathcal{U}(\frak{g})$
can be contracted to operators generating a copy of the same
semisimple algebra in the enveloping algebra of the contraction
$\frak{g}^{\prime}$.

\section{Invariants of Lie algebras. Maurer-Cartan equations}

Given a Lie algebra $ \frak{g}=\left\{X_{1},..,X_{n}\; |\;
\left[X_{i},X_{j}\right]=C_{ij}^{k}X_{k}\right\}$ in terms of
generators and commutation relations, we are primarily interested
in (polynomial) operators
$C_{p}=\alpha^{i_{1}..i_{p}}X_{i_{1}}..X_{i_{p}}$ in the
generators of $\frak{s}$ such that the constraint $
\left[X_{i},C_{p}\right]=0$,\; ($i=1,..,n$) is satisfied. Such an
operator can be shown to lie in the centre of the enveloping
algebra of $\frak{g}$, and is traditionally referred to as Casimir
operator. For semisimple Lie algebras, the determination of
Casimir operators can be done using structural properties
\cite{Per}. However, for non-semisimple Lie algebras the relevant
invariant functions are often rational or even transcendental
functions \cite{Bo}. This suggest to develop a method in order to
cover arbitrary Lie algebras. One convenient approach is the
analytical realization. The generators of the Lie algebra
$\frak{s}$ are realized in the space $C^{\infty }\left(
\frak{g}^{\ast }\right) $ by means of the differential operators:
\begin{equation}
\widehat{X}_{i}=C_{ij}^{k}x_{k}\frac{\partial }{\partial x_{j}},
\label{Rep1}
\end{equation}
where $\left\{ x_{1},..,x_{n}\right\}$ is a dual basis of
$\left\{X_{1},..,X_{n}\right\} $. The invariants of $\frak{g}$ (in
particular, the Casimir operators) are solutions of the following
system of partial differential equations:
\begin{equation}
\widehat{X}_{i}F=0,\quad 1\leq i\leq n.  \label{sys}
\end{equation}
Whenever we have a polynomial solution of (\ref{sys}), the
symmetrization map defined by
\begin{equation}
Sym(x_{i_{1}}^{a_{1}}..x_{i_{p}}^{a_{p}})=\frac{1}{p!}\sum_{\sigma\in
S_{p}}x_{\sigma(i_{1})}^{a_{1}}..x_{\sigma(i_{p})}^{a_{p}}
\end{equation}
allows to recover the Casimir operators in their usual form, i.e,
as elements in the centre of the enveloping algebra of $\frak{g}$,
after replacing the variables $x_{i}$ by the corresponding
generator $X_{i}$ \cite{AA}. A maximal set of functionally
independent invariants is usually called a fundamental basis. The
number $\mathcal{N}(\frak{g})$ of functionally independent
solutions of (\ref{sys}) is obtained from the classical criteria
for differential equations, and is given by:
\begin{equation}
\mathcal{N}(\frak{g}):=\dim \,\frak{g}- {\rm
sup}_{x_{1},..,x_{n}}{\rm rank}\left( C_{ij}^{k}x_{k}\right),
\label{BB}
\end{equation}
where $A(\frak{g}):=\left(C_{ij}^{k}x_{k}\right)$ is the matrix
associated to the commutator table of $\frak{g}$ over the given
basis.

\medskip

The use of differential forms has turned out to be useful not only
to reformulate formula (\ref{BB}), but also to compute Casimir
operators in some situations \cite{C72}. In terms of the
Maurer-Cartan equations, the Lie algebra $\frak{g}$
is described as follows: Given the structure tensor $\left\{  C_{ij}%
^{k}\right\}  $ over the basis $\left\{  X_{1},..,X_{n}\right\} $,
the identification of the dual space $\frak{g}^{\ast}$ with the
left-invariant 1-forms on the simply connected Lie group whose
algebra is isomorphic to $\frak{g}$ allows to define an exterior
differential $d$ on $\frak{g}^{\ast}$ by
\begin{equation}
d\omega\left(  X_{i},X_{j}\right)  =-C_{ij}^{k}\omega\left(
X_{k}\right) ,\;\omega\in\frak{g}^{\ast}.\label{MCG}
\end{equation}
This coboundary operator $d$ allows us to rewrite $\frak{g}$ as a
closed system of $2$-forms%
\begin{equation}
d\omega_{k}=-C_{ij}^{k}\omega_{i}\wedge\omega_{j},\;1\leq
i<j\leq\dim\left( \frak{g}\right)  ,\label{MC2}
\end{equation}
called the Maurer-Cartan equations of $\frak{g}$. The closeness
condition $d^{2}\omega_{i}=0$ for all $i$ is equivalent to the
Jacobi condition. To reformulate equation (\ref{BB}) in terms of
differential forms, we consider the linear subspace
$\mathcal{L}(\frak{g})=\mathbb{R}\left\{ d\omega_{i}\right\}
_{1\leq i\leq \dim\frak{g}}$ of $\bigwedge^{2}\frak{g}^{\ast}$
generated by the $2$-forms $d\omega_{i}$ \cite{C43}. In
particular, $\dim\mathcal{L}(\frak{g})=\dim\left( \frak{g}\right)
$ is equivalent to the condition $\dim\left( \frak{g}\right)
=\dim\left[ \frak{g},\frak{g}\right] $. Considering a generic
element $\omega=a^{i}d\omega_{i}\,\;\left(
a^{i}\in\mathbb{R}\right)  $ of $\mathcal{L}(\frak{g})$, then we
can always find a positive integer $j_{0}\left( \omega\right)
\in\mathbb{N}$ such that $\bigwedge^{j_{0}\left( \omega\right)
}\omega\neq0$ and $\bigwedge ^{j_{0}\left( \omega\right)
+1}\omega\equiv0$. Thus defining the quantity $j_{0}\left(
\frak{g}\right) $ as the maximal rank of generic elements,
\begin{equation}
j_{0}\left(  \frak{g}\right)  =\max\left\{  j_{0}\left(
\omega\right) \;|\;\omega\in\mathcal{L}(\frak{g})\right\},
\label{MCa1}
\end{equation}
we get a numerical invariant of the Lie algebra $\frak{g}$ that
can be used to rewrite equation (\ref{BB}) as
\begin{equation}
\mathcal{N}\left(  \frak{g}\right)  =\dim\frak{g}-2j_{0}\left(  \frak{g}%
\right). \label{BB1}
\end{equation}
The scalar $j_{0}\left(  \frak{g}\right)$ further determines the
number of internal labels necessary to describe a general
irreducible representation of $\frak{g}$ \cite{Ra,Sh,C43}.

\medskip

\section{Virtual copies of semisimple Lie algebras}

Ch. Quesne developed in \cite{Que} a method based on enveloping
algebras to compute the Casimir operators of semidirect products
$\frak{s}\overrightarrow {\frak{\oplus}}_{R}\frak{w}\left(
n\right)  $ of simple Lie algebras $\frak{s}$ with Heisenberg-Weyl
algebras $\frak{w}\left(  n\right)  $ spanned by $n$ pairs of
boson operators $b_{i},b_{j}^{\dagger}$ and the unit operator
$\mathbb{I}$. Since the Casimir operators of simple Lie algebras
are known \cite{Per}, the method reduces essentially to find
polynomials in the generators of
$\frak{s}\overrightarrow{\frak{\oplus}}_{R}\frak{w}\left( n\right)
$ that span a copy of $\frak{s}$ in the enveloping algebra of
$\frak{s}\overrightarrow{\frak{\oplus}}_{R}\frak{w}\left(
n\right)  $. More specifically, taking the basis $\left\{
E_{ij},b_{i},b_{j}^{\dagger
},\mathbb{I}\right\}  $ of $\frak{s}\overrightarrow{\frak{\oplus}}_{R}%
\frak{w}\left(  n\right)  $, where the $E_{ij}$ spans the Levi part $\frak{s}%
$, the author constructed quadratic operators
\begin{equation}
E_{ij}^{\prime}=E_{ij}\mathbb{I}-b_{i}^{\dagger}b_{j}\label{Ope1}
\end{equation}
satisfying the constraints
\begin{equation}
\left[  E_{ij}^{\prime},b_{k}^{\dagger}\right]  =\left[  E_{ij}^{\prime}%
,b_{k}\right]  =0,\;\left[  E_{ij},E_{kl}^{\prime}\right]  =\left[
E_{ij},E_{kl}\right]  ^{\prime}. \label{EA1}%
\end{equation}
These identities ensure that the $E_{ij}^{\prime}$ transform in a
similar way to the original generators of $\frak{s}$, i.e.,
satisfy the commutation relation
\begin{equation}
\left[  E_{ij}^{\prime},E_{kl}^{\prime}\right]=\mathbb{I}\left[
E_{ij},E_{kl}\right]^{\prime}, \label{Oper2}
\end{equation}
and therefore span a copy of the latter
in the enveloping algebra of $\frak{s}\overrightarrow{\frak{\oplus}}%
_{R}\frak{w}\left(  n\right)  $.\footnote{Observe that, since
$\mathbb{I}$ is the identity operator, it can be skipped from
equation (\ref{Oper2}).} Now, using the expression for a Casimir
operator of order $p$ of $\frak{s}$,
\begin{equation}
C_{p}=\sum_{i_{1}..i_{p}}E_{i_{1}i_{2}}...E_{i_{p}i_{1}},
\end{equation}
the invariant for the corresponding copy $\frak{s}^{\prime}$ in the enveloping algebra $\mathcal{U}%
\left(  \frak{s}\overrightarrow{\frak{\oplus}}_{R}\frak{w}\left(
n\right) \right)  $ is simply
\begin{equation}
C_{p}^{\prime}=\sum_{i_{1}..i_{p}}E_{i_{1}i_{2}}^{\prime}...E_{i_{p}i_{1}%
}^{\prime}.
\end{equation}
As a consequence of equation (\ref{EA1}), this operator satisfies
the identities
\begin{equation}
\left[  C_{p}^{\prime},b_{k}^{\dagger}\right]  =\left[  C_{p}^{\prime}%
,b_{k}\right]  =\left[  C_{p}^{\prime},E_{ij}\right]  =0, \label{AE2}%
\end{equation}
proving that $C_{p}^{\prime}$ is a Casimir operator of the Lie
algebra $\frak{s}\overrightarrow{\frak{\oplus}}_{R}\frak{w}\left(
n\right)$. Independence of the invariants computed by this method
follows easily by direct verification.

\smallskip

In this case, the validity of the method relies heavily on the
specific structure of these semidirect
products.\footnote{Specifically, it uses that
the semidirect product $\frak{s}\overrightarrow{\frak{\oplus}}_{R}%
\frak{w}\left(  n\right)  $ has a one dimensional centre generated
by the identity operator $\mathbb{I}$.} One can however ask
whether a similar ansatz can be made for other types of semidirect
products, including Lie algebras whose radical is not necessarily
a nilpotent algebra. It seems reasonable to require that such a
Lie algebra $\frak{g}$ satisfies the inequality
$\mathcal{N}(\frak{g})\geq \mathcal{N}(\frak{s})$, where
$\frak{s}$ is the Levi subalgebra of $\frak{g}$, although it is
not ensured a priori that the generalization of (\ref{Ope1}) leads
to functionally independent invariants of $\frak{g}$.

\bigskip

The general idea to expand the ansatz of \cite{Que} to more wide
types of semidirect products is to replace the operators
(\ref{Ope1}) by more general functions of the generators, and then
analyze under which constraints these operators transform in
similar way to the identities (\ref{EA1}). To this extent, we
start from an unspecified non-semisimple Lie algebra $\frak{g}$
with Levi decomposition
$\frak{g}=\frak{s}\overrightarrow{\frak{\oplus}}_{R}\frak{r}$,
where $\frak{s}$ denotes the Levi subalgebra and $\frak{r}$ the
radical.\footnote{Recall that the radical is the maximal solvable
ideal of $\frak{g}$, while the Levi subalgebra is the maximal
semisimple subalgebra.} Let $\left\{
X_{1},..,X_{n},Y_{1},..,Y_{m}\right\} $ be a basis such that
$\left\{  X_{1},..,X_{n}\right\}  $ spans $\frak{s}$ and $\left\{
Y_{1},..,Y_{m}\right\}  $ spans $\frak{r}$. We further suppose
that the structure tensor in $\frak{s}$ is given by
\begin{equation}
\left[  X_{i},X_{j}\right]  =C_{ij}^{k}X_{k}.\label{ST}
\end{equation}
We now define operators $X_{i}^{\prime}$ in the enveloping algebra
of $\frak{g}$ by means of
\begin{equation}
X_{i}^{\prime}=X_{i}\,f\left(  Y_{1},..,Y_{m}\right) +P_{i}\left(
Y_{1},..,Y_{m}\right)  ,\label{OP1}%
\end{equation}
where $P_{i}$ is a homogeneous polynomial of degree $k$ and $f$ is
homogeneous of degree $k-1$. In order to generalize the method of
\cite{Que}, we require
the constraints%
\begin{eqnarray}
\left[  X_{i}^{\prime},Y_{k}\right]    & =0,\label{Bed1}\\
\left[  X_{i}^{\prime},X_{j}\right]    & =\left[
X_{i},X_{j}\right] ^{\prime}:=C_{ij}^{k}\left(
X_{k}f+P_{k}\right).\label{Bed2}
\end{eqnarray}
to be satisfied for all generators. We thus have to analyze the
conditions that the functions $f$ and $P_{i}$ must satisfy in
order for these equations to hold. After development, equation
(\ref{Bed1}) can be rewritten as
\begin{equation}
\left[  X_{i}^{\prime},Y_{j}\right]  =\left[  X_{i}f,Y_{j}\right]  +\left[  P_{i}%
,Y_{j}\right]  =X_{i}\left[  f,Y_{j}\right]  +\left[
X_{i},Y_{j}\right]
\,f+\left[  P_{i},Y_{j}\right]  .\label{Eq1}%
\end{equation}
Since the quantities involved are homogeneous in the variables of
$\frak{s}$ and $\frak{r}$, we can reorder the terms according to
their degree. Observe that $X_{i}\left[  f,Y_{j}\right]  $ is a
homogeneous polynomial of degree $k-1$ in the variables $\left\{
Y_{1},..,Y_{m}\right\}  $, while $\left[ X_{i},Y_{j}\right]
\,f+\left[  P_{i},Y_{j}\right]  $ is of degree $k$. Therefore the
terms must fulfill the following identities
\begin{eqnarray}
\left[  f,Y_{j}\right] =0,\nonumber \\
\left[  X_{i},Y_{j}\right]  \,f+\left[  P_{i},Y_{j}\right]
=0.\label{Eq1A}
\end{eqnarray}
From the first one we conclude that $f$ must be a Casimir operator
of the radical $\frak{r}$. The second condition shows the
obstruction to the $P_{i}^{\prime}s$ to be also invariants of
$\frak{r}$. Observe in particular that if $Y_{j}$ commutes with
the generator $X_{i}$ of $\frak{s}$, the it must commute with the
operator $X_{i}^{\prime}$. We now focus on equation (\ref{Bed2}).
Development of the latter leads to the identity
\begin{equation}
\left[  X_{i}^{\prime},X_{j}\right]  =\left[  X_{i},X_{j}\right]
\,f-X_{i}\left[  X_{j},f\right]  +\left[  P_{i},X_{j}\right]
.\label{Eq2}
\end{equation}
The polynomial $\left[  X_{i},X_{j}\right]  \,f-X_{i}\left[
X_{j},f\right] $ is homogeneous of degree $k-1$ in the variables
of $\frak{r}$ and degree one in the variables of $\frak{s}$, while
$\left[  P_{i},X_{j}\right] $ is of degree zero in the latter set
of variables. This implies that the system
\begin{eqnarray}
\left[  X_{i},X_{j}\right]  \,f-X_{i}\left[  X_{j},f\right]    =C_{ij}%
^{k}X_{k}f\nonumber \\
\left[  P_{i},X_{j}\right]    =C_{ij}^{k}P_{k}\label{Eq3}
\end{eqnarray}
must be satisfied for arbitrary indices $i,j$. Because of
(\ref{ST}), the first equation reduces to
\begin{equation}
X_{i}\left[  X_{j},f\right]  =0,
\end{equation}
and since this holds for all generators of $\frak{s}$, we deduce
that $f$ must be a Casimir operator of $\frak{g}$ that depends
only on the variables of the radical $\frak{r}$. The second
condition tells that the $P_{i}^{\prime}s$ transform under the
$X_{j}^{\prime}s$ like the generators of the semisimple part
$\frak{s}$.

Taken together, it follows from these relations that the operators
$X_{i}^{\prime}$ transform like the generators of $\frak{s}$, up
to the additional factor $f$ given by a Casimir operator. More
precisely, for any indice $i,j$ we have the commutator
\begin{eqnarray}
\left[  X_{i}^{\prime},X_{j}^{\prime}\right]    & =\left[  X_{i}f+P_{i}%
,X_{j}f+P_{j}\right]  =\left[  X_{i}f+P_{i},X_{j}f\right]  +\left[
X_{i}f+P_{i},P_{j}\right] \nonumber \\
& =C_{ij}^{k}X_{k}f^{2}+C_{ij}^{k}P_{k}f+\left[
X_{i}^{\prime},P_{j}\right]
=f\left[  X_{i},X_{j}\right]  ^{\prime},%
\end{eqnarray}
where the last equality is a consequence of (\ref{Bed1}).

\medskip
Because of this factor $f$, we cannot in general find operators $X_{i}%
^{\prime\prime}$ in the enveloping algebra\footnote{The operators
$X_{i}^{\prime\prime}=X_{i}+\frac{1}{f}P_{i}$ would belong to the
fraction field of the enveloping algebra $\mathcal{U}\left(
\frak{g}\right)  $ \cite{AA}.} of $\frak{g}$ satisfying
(\ref{Bed1}) and (\ref{Bed2}) and such that
\[
\left[  X_{i}^{\prime\prime},X_{j}^{\prime\prime}\right]  =\left[  X_{i}%
,X_{j}\right]  ^{\prime\prime}.
\]

For this reason we will say that the operators $X_{i}^{\prime}$
generate a virtual copy of $\frak{s}$ in the enveloping algebra of
$\frak{g}$. We remark that only for the case where $f$ can be
identified with an identity operator, as happens for Heisenberg
radical, this virtual copy actually generates a copy in
$\mathcal{U}\left(  \frak{g}\right)  $. However, this functional
factor $f$ does not affect the validity of the method to compute
Casimir operators of $\frak{g}$ from those of the Levi part
$\frak{s}$.

\begin{theorem}
Let $\frak{s}$ be the Levi part of a Lie algebra $\frak{g}$ and
let $X_{i}^{\prime}=X_{i}\,f\left(  Y_{1},..,Y_{m}\right)
+P_{i}\left( Y_{1},..,Y_{m}\right)  $ be polynomials in the
generators of $\frak{g}$ satisfying equations (\ref{Eq1A}) and
(\ref{Eq3}). If $C=\sum\alpha ^{i_{1}..i_{p}}X_{i_{1}}..X_{i_{p}}$
is a Casimir operator of degree $p$ of $\frak{s}$, then
$C^{\prime}=\sum\alpha^{i_{1}..i_{p}}X_{i_{1}}^{\prime
}..X_{i_{p}}^{\prime}$ is a Casimir operator of $\frak{g}$ of
degree $p\deg f$. Moreover, $\mathcal{N}\left(  \frak{g}\right)
\geq\mathcal{N}\left( \frak{s}\right)  +1.$
\end{theorem}

\begin{proof}
It is clear from equations (\ref{Bed1}) and (\ref{Bed2}) that if
$C=\sum
\alpha^{i_{1}..i_{p}}X_{i_{1}}..X_{i_{p}}$ is a Casimir operator of $\frak{s}%
$, then $C^{\prime}=\sum\alpha^{i_{1}..i_{p}}X_{i_{1}}^{\prime}..X_{i_{p}%
}^{\prime}$ is a Casimir operator of $\frak{g}$. To prove the
second assertion, we observe that, since
$X_{i}^{\prime}=X_{i}f+P_{i}$, we can
rewrite $C^{\prime}$ as%
\begin{equation}
C^{\prime}=\sum\alpha^{i_{1}..i_{p}}\left(  X_{i_{1}}..X_{i_{p}}\,f^{p}%
+\sum_{t=1}^{p}\sum_{j_{1},..,j_{t}=i_{1}}^{i_{p}}X_{i_{1}}..P_{j_{1}%
}..P_{j_{t}}..X_{i_{p}}\,f^{p-t}\right)  . \label{ZerC}%
\end{equation}
For any $1\leq t\leq p$, the polynomial
\[
F_{\left[  j_{1},,.j_{t}\right]  }=\sum_{j_{1},..,j_{t}=i_{1}}^{i_{p}}%
X_{i_{1}}..P_{j_{1}}..P_{j_{t}}..X_{i_{p}}\,f^{p-t}%
\]
is homogeneous of degree $p-t$ in the generators $\left\{  X_{1}%
,..,X_{n}\right\}  $. In particular, $C_{p}^{\prime}$ contains
$C_{p}$ as the term of maximal degree in the latter set of
generators. Let $\left\{ C_{1},..,C_{l}\right\}  $ be a set of
independent Casimir operators of $\frak{s}$. We can thus find
indices $k_{1},..,k_{l}$ such that the associated Jacobian does
not vanish:
\begin{equation}
\frac{\partial\left\{  C_{1},..,C_{l}\right\}  }{\partial\left\{  X_{k_{1}%
},..,X_{k_{l}}\right\}  }\neq0. \label{Jac}%
\end{equation}
Let $d_{i}=\deg C_{i}$ denote the degree of each invariant. We now
rewrite $\left\{  C_{1}^{\prime},..,C_{l}^{\prime}\right\}  $
according to (\ref{ZerC}), and compute the Jacobian with respect
to the generators $\left\{  X_{k_{1}},..,X_{k_{l}}\right\}  $. We
explicitly obtain
\begin{equation}
\fl \frac{\partial\left\{  C_{1}^{\prime},..,C_{l}^{\prime}\right\}  }%
{\partial\left\{  X_{k_{1}},..,X_{k_{l}}\right\}  }=\det\left(
\begin{array}
[c]{ccc}%
\frac{\partial
C_{1}}{X_{k_{1}}}f^{d_{1}}+\sum_{t=1}^{d_{1}}\frac{\partial
}{\partial X_{k_{1}}}F_{1,[j_{1},,.j_{t}]} & .. & \frac{\partial C_{1}%
}{X_{k_{l}}}f^{d_{1}}+\sum_{t=1}^{d_{1}}\frac{\partial}{\partial X_{k_{l}}%
}F_{1,[j_{1},,.j_{t}]}\\
\vdots &  & \vdots\\
\frac{\partial
C_{l}}{X_{k_{1}}}f^{d_{l}}+\sum_{t=1}^{d_{l}}\frac{\partial
}{\partial X_{k_{1}}}F_{l,[j_{1},,.j_{t}]} & .. & \frac{\partial C_{l}%
}{X_{k_{l}}}f^{d_{l}}+\sum_{t=1}^{d_{l}}\frac{\partial}{\partial X_{k_{l}}%
}F_{l,[j_{1},,.j_{t}]}%
\end{array}
\right)  . \label{Det1}%
\end{equation}
Using the properties of determinants, it can be easily seen that
(\ref{Det1}) simplifies to
\begin{equation}
\fl \frac{\partial\left\{  C_{1}^{\prime},..,C_{l}^{\prime}\right\}  }%
{\partial\left\{  X_{k_{1}},..,X_{k_{l}}\right\}  }=\det\left(
\begin{array}
[c]{ccc}%
\frac{\partial C_{1}}{X_{k_{1}}} & .. & \frac{\partial C_{1}}{X_{k_{l}}}\\
\vdots &  & \vdots\\
\frac{\partial C_{l}}{X_{k_{1}}} & .. & \frac{\partial C_{l}}{X_{k_{l}}}%
\end{array}
\right)  f^{d_{1}+..+d_{l}}+\det\left(
\begin{array}
[c]{ccc}%
\sum_{t=1}^{d_{1}}\frac{\partial}{\partial
X_{k_{1}}}F_{1,[j_{1},,.j_{t}]} &
.. & \sum_{t=1}^{d_{1}}\frac{\partial}{\partial X_{k_{l}}}F_{1,[j_{1}%
,,.j_{t}]}\\
\vdots &  & \\
\sum_{t=1}^{d_{l}}\frac{\partial}{\partial
X_{k_{1}}}F_{l,[j_{1},,.j_{t}]} &
.. & +\sum_{t=1}^{d_{l}}\frac{\partial}{\partial X_{k_{l}}}F_{l,[j_{1}%
,,.j_{t}]}%
\end{array}
\right)  \label{Det2}%
\end{equation}
Since the first term of the latter equation coincides with $\frac
{\partial\left\{  C_{1},..,C_{l}\right\}  }{\partial\left\{  X_{k_{1}%
},..,X_{k_{l}}\right\}  }f^{d_{1}+..+d_{l}}$ and has maximal
degree in the $X_{i}^{\prime}$s, condition (\ref{Jac}) ensures
that (\ref{Det1}) does not vanish, i.e., the operators $\left\{
C_{1}^{\prime},..,C_{l}^{\prime }\right\}  $ are independent.
Finally, $f$ being a Casimir operator of $\frak{g}$ depending only
on the variables of the radical, we get $\mathcal{N}\left(
\frak{s}\right)  +1\leq\mathcal{N}\left(  \frak{g}\right) $.
\end{proof}

\medskip

Hence the independence of the invariants constructed starting from
those of the Levi part $\frak{s}$ is automatically ensured by the
conditions (\ref{Bed1}) and (\ref{Bed2}). This result further
allows us to obtain a criterion concerning the existence of
virtual copies of semisimple algebras in enveloping algebras of
semidirect products. The proof is a direct consequence of the
previous argument.

\begin{proposition}
Let $\frak{s}$ be the Levi subalgebra of the Lie algebra
$\frak{g}$. If $\mathcal{N}(\frak{g})\leq \mathcal{N}(\frak{s})$,
then no virtual copy of $\frak{s}$  in the enveloping algebra of
$\frak{g}$ exists.
\end{proposition}

This means specifically that for semidirect products with less
invariants than those of its Levi part, the latter provides no
valid information to simplify the computation of the Casimir
invariants. In these cases, an approach by enveloping algebras is
specially difficult, since the operators have to be found
directly. The ansatz further allows to establish some restrictions
on the structure of the radical $\frak{r}$ of such semidirect
products.

\begin{proposition}
Let $\frak{g}=\frak{s}\overrightarrow{\oplus}_{R}\frak{r}$ admit a
(virtual) copy of $\frak{s}$ in the enveloping algebra
$\mathcal{U}\left( \frak{g}\right)  $ generated by operators of
type (\ref{OP1}). Then the radical $\frak{r}$ is not Abelian.
\end{proposition}

The proof is a direct consequence of equation (\ref{Eq1A}). In
fact, if the operators $X_{ij}^{\prime}=X_{i}^{\prime}\,f\left(
Y_{1},..,Y_{m}\right) +P_{i}\left(  Y_{1},..,Y_{m}\right)  $
satisfy the constraint (\ref{Bed1}), by the second equation of
(\ref{Eq1A}) we have that
\[
\left[  X_{i},Y_{j}\right]  f+\left[  P_{i},Y_{j}\right]  =0.
\]
If $\frak{r}$ were Abelian, then obviously  $\left[
P_{i},Y_{j}\right]  =0$ for all indices $i,j$, and this would
imply, since $f\neq0$, that $\left[  X_{i},Y_{j}\right]  =0$ for
all generators $X_{i},Y_{j}$. But this tells that the
representation $R$ of $\frak{s}$ on $\frak{r}$ is trivial, hence
the algebra $\frak{g}$ reduces to a direct sum
$\frak{s}\oplus\frak{r}$.

\section{Invariants of the inhomogeneous Hamilton algebra}

The Hamilton algebra and its various extensions arise naturally in
the analysis of relativity groups for noninertial states in the
quantum mechanical frame \cite{SL1}. It is known that projective
representations on physical states lead to the conclusion that
relativity groups arise as subgroups of the automorphism group of
the Heisenberg algebra $\frak{h}_{N}$. In particular, it follows
that the Hamilton group, which is isomorphic to the semidirect
product of $\frak{so}(N)$ with the Heisenberg algebra
$\frak{h}_{N}$, coincides with the relativity group for
noninertial frames in classical Hamilton mechanics \cite{SL2}. Its
Lie algebra $Ha\left( N\right) $ is given, over the basis $\left\{
J_{ij},G_{k},F_{k},R\right\}$, by the brackets
\begin{equation}
\begin{array}
[c]{l}%
\left[  J_{ij},J_{kl}\right]  =\delta_{i}^{l}J_{jk}+\delta_{j}^{k}%
J_{il}-\delta_{j}^{l}J_{ik}-\delta_{i}^{k}J_{jl},\quad \left[
G_{i},F_{j}\right]
=\delta_{i}^{j}R.\\
\left[  J_{ij},G_{k}\right]  =-\delta_{i}^{k}G_{j}+\delta_{k}^{j}%
G_{i},\quad \left[  J_{ij},F_{k}\right]  =-\delta_{i}^{k}F_{j}+\delta_{k}^{j}%
F_{i},
\end{array}\label{Kl1}
\end{equation}
The action of $\frak{so}(N)$ over the Heisenberg radical is easily
seen to be described by the fundamental $N$-dimensional tensor
representation $\Lambda$ of $\frak{so}\left( N\right) $. Therefore
its dimension is $\dim Ha\left(  N\right) =\frac{1}{2}N\left(
N+3\right)+1$. Following \cite{Que}, this algebra has $\left[
\frac{N}{2}\right] +1$ Casimir operators. In this case, the
operators
\[
J_{ij}^{\prime}=J_{ij}R+G_{i}F_{j}-G_{j}F_{i}
\]
generate the searched copy of $\frak{so}\left(  N\right)$. Thus
the Casimir operators of $Ha(N)$ are obtained as the symmetrized
polynomials of the functions $C_{k}$ arising as coefficients of
the polynomial
\begin{equation*}
Cas_{N}=\det\left|  A_{N}-T\mathrm{Id}_{N}\right|  =\sum_{l=1}^{N}%
C_{2l}T^{2N-2l}+T^{2N},
\end{equation*}
where $A_{N}$ is the $N\times N$ matrix defined by
\begin{equation}
A_{N}=\left(
\begin{array}
[c]{cccc}%
0 & J_{12}^{\prime} & ... & J_{1N}^{\prime}\\
-J_{12}^{\prime} & 0 & ... & J_{2N}^{\prime}\\
\vdots & \vdots &  & \vdots\\
-J_{1N}^{\prime} & -J_{2N}^{\prime} & ... & 0
\end{array}
\right)  . \label{SON}
\end{equation}

An immediate question that arises naturally from this is whether
the various extensions of the Hamilton algebra \cite{SL1,SL2} can
be analyzed by means of virtual copies of $\frak{so}(N)$ in the
corresponding enveloping algebras. In particular, insertion of the
operators generating the corresponding copies into matrix
(\ref{SON}) would provide the fundamental set of invariants for
each of these extensions for arbitrary dimension $N$.

\medskip

We first consider the inhomogeneous Hamilton algebra $IHa(N)$,
which is the first step in the study of quantum noninertial states
\cite{SL3}. This (non-central) extension is obtained by addition
of the generators $\left\{ Q_{k},P_{k},E,T\right\} $ to those of
$Ha(N)$: The brackets are those of (\ref{Kl1}), to which the
following are added:
\begin{equation}
\fl \begin{tabular}
[c]{l}%
$\left[  J_{ij},Q_{k}\right]  =-\delta_{i}^{k}Q_{j}+\delta_{k}^{j}%
Q_{i},\quad \left[  J_{ij},P_{k}\right]  =-\delta_{i}^{k}P_{j}+\delta_{k}^{j}%
P_{i},\quad \left[  G_{i},Q_{j}\right]  =\delta_{i}^{j}T,$\\
$\left[  F_{i},P_{j}\right]  =\delta_{i}^{j}T,\quad \left[
E,G_{i}\right]=-P_{i},\quad \left[  E,F_{i}\right]  =Q_{i},\quad \left[  E,R\right]  =2T.$%
\end{tabular}\label{Kl2}
\end{equation}
It is observed that $R$ is no more a central element, but now $T$
plays the role of central charge. In particular, the radical is a
solvable non-nilpotent Lie algebra. The first question that arises
is how many invariants this algebra $IHa\left(  N\right)$ has.

\begin{lemma}
For any $N\geq 3$ following identity holds:
\begin{equation}
\mathcal{N}\left(IHa\left(  N\right)  \right) =\left[
\frac{N}{2}\right]  +1.
\end{equation}
\end{lemma}

\begin{proof}
We consider the Maurer-Cartan equations (MC) of the algebra to
prove it \cite{C43}. Let $\left\{
\omega_{ij},\theta_{k},\eta_{k},\varphi_{k},\chi
_{k},\omega_{R},\omega_{E},\omega_{T}\right\}  $ be the dual basis
to $\left\{  J_{ij},G_{k},F_{k},Q_{k},P_{k},R,E,T\right\}  $. The
Maurer-Cartan
equations are easily seen to be%
\begin{equation}
\begin{tabular}
[c]{l}%
$d\omega_{ij}=-\sum_{k=1}^{N}\omega_{ik}\wedge\omega_{jk},\;d\theta_{k}%
=\sum_{j=1}^{N}\omega_{kj}\wedge\theta_{j},\;d\eta_{k}=\sum_{j=1}^{N}%
\omega_{kj}\wedge\eta_{j},$\\
$d\varphi_{k}=\sum_{j=1}^{N}\omega_{kj}\wedge\varphi_{j}+\omega_{E}\wedge
\eta_{k},\;\;d\chi_{k}=\sum_{j=1}^{N}\omega_{kj}\wedge\chi_{j}-\omega
_{E}\wedge\theta_{k},$\\
$d\omega_{R}=\sum_{k=1}^{N}\theta_{k}\wedge\eta_{k},\;d\omega_{T}=\sum
_{k=1}^{N}\left(
\theta_{k}\wedge\varphi_{k}+\eta_{k}\wedge\chi_{k}\right)
+2\omega_{E}\wedge\omega_{R},$\\
$d\omega_{E}=0.$%
\end{tabular}
\end{equation}
Any generic element $\omega\in\mathcal{L}\left(  \frak{g}\right)
$:
\begin{equation}
\omega=a^{ij}d\omega_{ij}+b_{1}^{k}d\theta_{k}+b_{2}^{k}d\eta_{k}+b_{3}%
^{k}d\varphi_{k}+b_{4}^{k}d\chi_{k}+b_{5}d\omega_{R}+b_{6}d\omega_{T},
\end{equation}
where $\;a^{ij},b_{j}^{k}\in\mathbb{R}$ are arbitrary constants,
can be decomposed as $\omega=\xi_{1}+\xi_{2}$, where $\xi_{1}=$
$a^{ij}d\omega
_{ij}\in\mathcal{L}\left(  \frak{s}\right)  $ and $\xi_{2}=b_{1}^{k}%
d\theta_{k}+b_{2}^{k}d\eta_{k}+b_{3}^{k}d\varphi_{k}+b_{4}^{k}d\chi_{k}%
+b_{5}d\omega_{R}+b_{6}d\omega_{T}\in\mathcal{L}\left(
\frak{r}\right)  $. A routine but tedious computation shows that
\begin{equation}
\bigwedge^{p}\left(  \xi_{2}-b_{6}d\omega_{T}\right)
=0,\;p\geq2N.\label{PLR}
\end{equation}
This equation follows essentially from the fact that
$\bigwedge^{N} d\theta_{k}=\bigwedge^{N} d\eta_{k}=\bigwedge^{N+1}
d\varphi_{k}=\bigwedge^{N+1} d\chi _{k}=0$ and that the terms in
the MC equations of these generators involve all rotation 1-forms
$\omega_{ij}$. We further observe that $d\omega_{T}$ involves all
1-forms associated to generators of the radical with the exception
of $T$, and satisfies the
following identity%
\begin{equation}
\fl\bigwedge^{2N+1}d\omega_{T}=\pm2\left(  2N+1\right)  !\theta_{1}%
\wedge..\wedge\theta_{N}\wedge\varphi_{1}\wedge..\wedge\varphi_{N}\wedge
\chi_{1}\wedge..\wedge\chi_{N}\wedge\eta_{1}\wedge..\wedge\eta_{N}\wedge
\omega_{R}\wedge\omega_{E}.\label{Me1}%
\end{equation}
This proves that the equality $j_{0}\left(  \xi_{2}\right) =2N+1$
is satisfied only if $b_{6}\neq0$. For this reason we fix
$\xi_{2}=d\omega_{T}$. This choice has the advantage that no
element of the dual space to $\frak{so}\left( N\right)$ appears in
the expression of $\xi_{2}$. We now consider the Maurer-Cartan
equations of the Levi part. Since $j_{0}\left( \frak{so}\left(
N\right) \right) =\frac{1}{2}\left( \frac{N\left( N-1\right)
}{2}+\left[ \frac{N}{2}\right]  \right)  ,$ there exist
non-vanishing coefficients $\alpha^{ij}$ such that
\begin{equation}
\bigwedge^{j_{0}\left(  \frak{so}\left(  N\right)  \right)  }\xi_{1}%
=\bigwedge^{j_{0}\left(  \frak{so}\left(  N\right)  \right)
}\left( a^{ij}d\omega_{ij}\right)  \neq0.\label{Me2}
\end{equation}
Taking the 2-form $\xi_{1}+d\omega_{T}$,  the wedge products can
be rewritten as a sum
\begin{equation}
\bigwedge^{p}\left(  \xi_{1}+d\omega_{T}\right)
=\sum_{k=1}^{p}\left(
\begin{array}
[c]{c}%
p\\
k
\end{array}
\right)  \left(  \bigwedge^{k}\xi_{1}\right)  \wedge\left(
\bigwedge ^{p-k}d\omega_{T}\right)  .
\end{equation}
The conditions (\ref{Me1}) and (\ref{Me2}) imply that $k\leq
j_{0}\left( \frak{so}\left( N\right)  \right)  $ and
$p-k\leq2N+1$, respectively. This means that $j_{0}\left(
IHa\left( N\right) \right) \leq j_{0}\left( \frak{so}\left(
N\right) \right) +2N+1$. On the other hand, since $\xi_{1}$ and
$d\omega_{T}$ have no terms in common, it follows at once that
\begin{equation}
\bigwedge^{j_{0}\left(  \frak{so}\left(  N\right)  \right)
+2N+1}\xi _{1}+d\omega_{T}\neq0,
\end{equation}
showing that
\begin{equation}
j_{0}\left(  IHa\left(  N\right)  \right)  =\frac{N\left(  N-1\right)  }%
{2}-\frac{1}{2}\left[  \frac{N}{2}\right]  +2N+1.
\end{equation}
Applying formula (\ref{BB1}) we finally get
\begin{equation}
\mathcal{N}\left(  IHa\left(  N\right)  \right)  =\left[
\frac{N}{2}\right] +1.\label{ZdI1}
\end{equation}
The latter equation shows that adding the new generators $\left\{  Q_{k}%
,P_{k},E,T\right\}$ has no consequence on the number of
invariants.
\end{proof}

We conclude that the conditions to apply the generalized ansatz
developed earlier are satisfied. Observe that since $T$ is a
central element of $IHa(N)$ (and the only invariant of the
radical), a power of it must be the function $f$ multiplying the
rotation generators $J_{ij}$ in the operators generating the
virtual copy.

\begin{proposition}
For any $N\geq3$, the operators%
\begin{equation}
\fl J_{ij}^{\prime}=J_{ij}T^{2}+\left(
G_{i}Q_{j}-G_{j}Q_{i}\right) T+\left( F_{i}P_{j}-F_{j}P_{i}\right)
T+\left( P_{i}Q_{j}-P_{j}Q_{i}\right) R
\end{equation}
generate a virtual copy of $\frak{so}\left( N\right)  $ in the
enveloping algebra of $IHa\left(  N\right)  $. The Casimir
operators of the latter are obtained as the symmetrized
polynomials of the functions $C_{k}$ arising as coefficients of
the characteristic
polynomial%
\[
Cas_{N}=\det\left|  A_{N}-T\mathrm{Id}_{N}\right|  =\sum_{l=1}^{N}%
C_{2l}T^{2N-2l}+T^{2N},
\]
where $A_{N}$ is the matrix (\ref{SON}).
\end{proposition}

The proof is straightforward and follows by direct computation.
From (\ref{Kl1}) and (\ref{Kl2}) we get that $R,T$ commute with
all generators but $E$, which does not appear in the expression of
the operators $J_{ij}^{\prime}$. To verify equations (\ref{Bed1})
and (\ref{Bed2}), it suffices to compute the commutators of $\
$the basis elements $\left\{
J_{ij},G_{k},F_{k},Q_{k},P_{k},R,E,T\right\}  $ with the quadratic
polynomials  $\left(  G_{i}Q_{j}-G_{j}Q_{i}\right)  $, $\left(  F_{i}%
P_{j}-F_{j}P_{i}\right)  $ and $\left(
P_{i}Q_{j}-P_{j}Q_{i}\right)  $. We omit the explicit
computations, which are mechanical but quite lengthy. The main
commutators to reproduce these commutators are given in Appendix
A.

\subsection{Invariants of the centrally extended Hamilton algebras}

According to the quantum formulation, the realization of physical
states and the projective representations of relativity groups are
equivalent to unitary representations of central extensions.
Therefore the interesting object to be studied in the context of
noninertial states are the central extension of the inhomogeneous
Hamilton group \cite{SL2}. It can be easily proved by means of
cohomological tools \cite{Az2} that the inhomogeneous Hamilton
algebra $IHa(N)$ admits a three dimensional central extension
denoted by $QHa\left( N\right)$. The new central generators
$L,M,A$ are related to time-energy and position momentum relation,
the mass generator and reciprocal symmetry, respectively.
$QHa\left( N\right)$ is the maximal centrally extended Lie algebra
obtained from $IHa(N)$. Taking a basis $\left\{
J_{ij},G_{k},F_{k},P_{k,}Q_{k},R,E,T,L,A,M\right\}$, where $L,A,M$
are the generators of the (maximal) central extension, the
commutation relations are given by
\begin{equation}
\fl
\begin{tabular}
[c]{l}
$\left[  J_{ij},J_{kl}\right]  =\delta_{i}^{l}J_{jk}+\delta_{j}^{k}%
J_{il}-\delta_{j}^{l}J_{ik}-\delta_{i}^{k}J_{jl},\;\left[  J_{ij}%
,G_{k}\right]  =-\delta_{i}^{k}G_{j}+\delta_{k}^{j}G_{i},$\\
$\left[  J_{ij},F_{k}\right]  =-\delta_{i}^{k}F_{j}+\delta_{k}^{j}%
F_{i},\;\left[  J_{ij},P_{k}\right]  =-\delta_{i}^{k}P_{j}+\delta_{k}^{j}%
P_{i},\;\left[  J_{ij},Q_{k}\right]  =-\delta_{i}^{k}Q_{j}+\delta_{k}^{j}%
Q_{i},$\\
$\left[  G_{i},F_{j}\right]  =\delta_{i}^{j}R,\;\left[
G_{i},Q_{j}\right] =\delta_{i}^{j}T,\;\left[  F_{i},P_{j}\right]
=\delta_{i}^{j}T,\;\left[ P_{i},Q_{j}\right]
=\delta_{i}^{j}L,\;\left[  G_{i},P_{j}\right]  =\delta
_{i}^{j}M,$\\
$\left[  F_{i},Q_{j}\right]  =\delta_{i}^{j}A,\;\left[
E,G_{i}\right] =-P_{i},\;\left[  E,F_{i}\right]  =Q_{i},\;\left[
E,T\right]  =-L,\;\left[E,R\right]  =2T.$
\end{tabular}
\end{equation}

Again, the use of differential forms allows us to compute the
number of invariants for arbitrary values of $N$.

\begin{lemma}
For any $N\geq3$ the following equality holds%
\[
\mathcal{N}\left( QHa\left( N\right) \right)  =\left[
\frac{N}{2}\right] +4.
\]
\end{lemma}

\begin{proof}
The proof is essentially a consequence of the argument in Lemma 1.
For $QHa\left(  N\right)  $ the Maurer-Cartan equations are%
\begin{equation}
\begin{tabular}
[c]{l}%
$d\omega_{ij}=-\sum_{k=1}^{N}\omega_{ik}\wedge\omega_{jk},\;d\theta_{k}%
=\sum_{j=1}^{N}\omega_{kj}\wedge\theta_{j},\;d\eta_{k}=\sum_{j=1}^{N}%
\omega_{kj}\wedge\eta_{j},$\\
$d\varphi_{k}=\sum_{j=1}^{N}\omega_{kj}\wedge\varphi_{j}+\omega_{E}\wedge
\eta_{k},\;\;d\chi_{k}=\sum_{j=1}^{N}\omega_{kj}\wedge\chi_{j}-\omega
_{E}\wedge\theta_{k},$\\
$d\omega_{R}=\sum_{k=1}^{N}\theta_{k}\wedge\eta_{k},\;d\omega_{T}=\sum
_{k=1}^{N}\left(
\theta_{k}\wedge\varphi_{k}+\eta_{k}\wedge\chi_{k}\right)
+2\omega_{E}\wedge\omega_{R},$\\
$d\omega_{L}=\sum_{k=1}^{N}\chi_{k}\wedge\varphi_{k}-\omega_{E}\wedge
\omega_{T},\;d\omega_{M}=\sum_{k=1}^{N}\theta_{k}\wedge\chi_{k},$\\
$d\omega_{A}=\sum_{k=1}^{N}\eta_{k}\wedge\varphi_{k},\;d\omega_{E}=0.$%
\end{tabular}\label{MCGM}
\end{equation}
Since $\left\{  L,M,A\right\}  $ are central elements, the forms
$\omega _{A},\omega_{M}$ and $\omega_{L}$ do not appear in the
terms of a generic element of $\mathcal{L}\left(  QHa\left(
N\right)  \right)  $. Therefore, the same 2-form used in Lemma 1
provides the maximal rank of an element in $\mathcal{L}\left(
QHa\left(  N\right)  \right)  $. We thus conclude that
$j_{0}\left(  QHa\left(  N\right)  \right)  =\frac{1}{2}\left(
\frac{N\left( N-1\right)  }{2}-\left[  \frac{N}{2}\right]
+4N+2\right)  $ and again, by
formula (\ref{BB1}), we obtain the number of invariants%
\begin{equation*}
\fl \mathcal{N}\left(  QHa\left(  N\right)  \right)  =\frac{N^{2}+7N+12}%
{2}-\left(  \frac{N\left(  N-1\right)  }{2}-\left[
\frac{N}{2}\right] +4N+2\right)  =\left[  \frac{N}{2}\right]  +4.
\end{equation*}
\end{proof}

The problem therefore reduces to compute the $\left[
\frac{N}{2}\right]+1$ non-central invariants. Observe that,
according to Theorem 1, one of these invariants must be an
invariant depending only on the generators of the radical.
Moreover, in contrast to the invariants of the inhomogeneous
Hamilton algebra, which depended upon all its generators, for the
quantum Hamilton $QHa\left(  N\right)$ algebra we find an
important property that turns out to be essential to find closed
formulae for the invariants.

\begin{lemma}
For any generalized Casimir invariant $\Phi$ of QHa$\left(
N\right)  $ the
constraint%
\[
\frac{\partial\Phi}{\partial e}=0
\]
is satisfied.
\end{lemma}

It follows at once from the representation (\ref{sys}) that the
differential operator associated to the generator $R$ is given by
\[
\widehat{R}=-2t\frac{\partial}{\partial e},
\]
thus the invariants do not depend on the variable $e$. An
important consequence of this fact is that the generators
$\left\{T,R,M,L,A\right\}$ commute with all other remaining
generators $J_{ij},G_{k},F_{k},Q_{k},P_{k}$. This will be crucial
in finding new operators in the enveloping algebra that generate a
virtual copy of $\frak{so}(N)$.

\begin{proposition}
For any $N\geq3$, the operators%
\begin{eqnarray}
\fl \widehat{J}_{ij}= J_{ij}\left(  T^{2}%
+RL-AM\right)  +T\left(  G_{i}Q_{j}-G_{j}Q_{i}\right)  +T\left(  F_{i}%
P_{j}-F_{j}P_{i}\right)  +L\left(  G_{i}F_{j}-G_{j}F_{i}\right)  +\nonumber \\
  +R\left(  P_{i}Q_{j}-P_{j}Q_{i}\right)  +M\left(  Q_{i}F_{j}%
-Q_{j}F_{i}\right)  +A\left( P_{i}G_{j}-P_{j}G_{i}\right),
\label{NE3}
\end{eqnarray}
generate a virtual copy of $\frak{so}\left(  N\right)  $ in the
enveloping algebra of $QHa\left(  N\right)  $.
\end{proposition}

The proof, once more, follows by direct verification of the
conditions (\ref{Bed1}) and (\ref{Bed2}). As before, we omit the
long verification. The main commutators of quadratic polynomials
in the generators are given in Appendix A.\newline We concluide
that a fundamental set of invariants of $QHa(N)$ for arbitrary $N$
is given by $\left\{L,A,M,T^2+RN-AM,Sym(C_{2l}^{\prime})\right\}$,
where $Sym(C_{2l}^{\prime})$ denotes the symmetrization of the
polynomials $C_{2l}$ obtained inserting the operators (\ref{NE3})
into the matrix (\ref{SON}) and computing the characteristic
polynomial.

\medskip
We already observed that $QHa(N)$ is a maximal central extension
of $IHa(N)$. The remaining extensions are either one dimensional,
with central generator $L$,$A$ or $M$, or two dimensional, with
central charges $\left\{L,A\right\}$, $\left\{L,M\right\}$ or
$\left\{A,M\right\}$. In view of the preceding results, it seems
reasonable to think that slight modifications of (\ref{NE3}) will
also provide the corresponding operators that generate the virtual
copy of $\frak{so}(N)$ in the enveloping algebra of these
extensions. The Maurer-Cartan equations can be easily obtained
from (\ref{MCGM}), and using them it can be easily shown that
$\mathcal{N}(IHa(N)\uplus \left\langle
\Lambda\right\rangle)=\left[\frac{N}{2}\right]+2$ and
$\mathcal{N}(IHa(N)\uplus \left\langle
\Lambda,\Pi\right\rangle)=\left[\frac{N}{2}\right]+3$ for
$\Lambda,\Pi\in\left\{A,M,L\right\}$. The operators generating the
virtual copies are given in Table 1.

\begin{table}
\caption{Generators of virtual copies of $\frak{so}(n)$ in the
enveloping algebras of one and two dimensional central extensions
of $IHa(N)$.}
\begin{tabular}
[c]{ll}%
$\frak{g}$ & Operators $\widehat{J}_{ij}$ generating the virtual
copy of $\frak{so}\left( N\right)  $ in $\frak{U}\left(
\frak{g}\right) $\\\br
$IHa\left(  n\right)  \uplus\left\langle L\right\rangle $ & $%
\begin{array}
[c]{l}%
 J_{ij}\left(  T^{2}+RL\right)  +T\left(  G_{i}Q_{j}%
-G_{j}Q_{i}+F_{i}P_{j}-F_{j}P_{i}\right)  +L\left(  G_{i}F_{j}-G_{j}%
F_{i}\right)  \\
 +R\left(  P_{i}Q_{j}-P_{j}Q_{i}\right)
\end{array}
\ $\\
$IHa\left(  n\right)  \uplus\left\langle M\right\rangle $ & $%
\begin{array}
[c]{l}%
 J_{ij}T^{2}+T\left(  G_{i}Q_{j}-G_{j}Q_{i}+F_{i}%
P_{j}-F_{j}P_{i}\right)  +M\left(  Q_{i}F_{j}-Q_{j}F_{i}\right)  \\
 +R\left(  P_{i}Q_{j}-P_{j}Q_{i}\right)
\end{array}
$\\
$IHa\left(  n\right)  \uplus\left\langle A\right\rangle $ & $%
\begin{array}
[c]{l}%
 J_{ij}T^{2}+T\left(  G_{i}Q_{j}-G_{j}Q_{i}+F_{i}%
P_{j}-F_{j}P_{i}\right)  +A\left(  P_{i}G_{j}-P_{j}G_{i}\right)  \\
 +R\left(  P_{i}Q_{j}-P_{j}Q_{i}\right)
\end{array}
$\\
$IHa\left(  n\right)  \uplus\left\langle A,M\right\rangle $ & $%
\begin{array}
[c]{l}%
 J_{ij}\left(  T^{2}-AM\right)  +T\left(  G_{i}Q_{j}%
-G_{j}Q_{i}+F_{i}P_{j}-F_{j}P_{i}\right)  +R\left(  P_{i}Q_{j}-P_{j}%
Q_{i}\right)  \\
 +M\left(  Q_{i}F_{j}-Q_{j}F_{i}\right)  +A\left(  P_{i}G_{j}-P_{j}%
G_{i}\right)
\end{array}
$\\
$IHa\left(  n\right)  \uplus\left\langle A,L\right\rangle $ & $%
\begin{array}
[c]{l}%
 J_{ij}\left(  T^{2}+RL\right)  +T\left(  G_{i}Q_{j}%
-G_{j}Q_{i}+F_{i}P_{j}-F_{j}P_{i}\right)  +L\left(  G_{i}F_{j}-G_{j}%
F_{i}\right)  \\
 +R\left(  P_{i}Q_{j}-P_{j}Q_{i}\right)  +A\left(  P_{i}G_{j}-P_{j}%
G_{i}\right)
\end{array}
$\\
$IHa\left(  n\right)  \uplus\left\langle L,M\right\rangle $ & $%
\begin{array}
[c]{l}%
 J_{ij}\left(  T^{2}+RL\right)  +T\left(  G_{i}Q_{j}%
-G_{j}Q_{i}+F_{i}P_{j}-F_{j}P_{i}\right)  +L\left(  G_{i}F_{j}-G_{j}%
F_{i}\right)  \\
 +R\left(  P_{i}Q_{j}-P_{j}Q_{i}\right)  +M\left(  Q_{i}F_{j}-Q_{j}%
F_{i}\right)
\end{array}
$\\\br
\end{tabular}
\end{table}

\section{Applications to contractions}

Considering that the proposed method is motivated by the problem
of finding the invariants of semidirect products, it is natural to
ask whether the procedure of virtual copies can be extended to
contractions of Lie algebras, in order to obtain a systematic
method to compute the invariants of the contraction without being
forced to consider the invariants separately \cite{Boh,We}. In
this section we outline a possible approach to the contraction
problem. Recall that a contraction of a Lie algebra $\frak{g}$ is
determined by a family of non-singular linear maps
$\Phi_{\varepsilon}$ of $\frak{g}$, where $\varepsilon \in (0,1]$.
If the limit
\begin{equation}
\left[X,Y\right]_{\infty}:=\lim_{\varepsilon\rightarrow
0}\Phi_{\varepsilon}^{-1}\left[\Phi_{\varepsilon}(X),\Phi_{\varepsilon}(Y)\right]
\label{Kont}
\end{equation}
exists for any $X,Y\in\frak{g}$, then equation (\ref{Kont})
defines a Lie algebra $\frak{g}^{\prime}$ called the contraction
of $\frak{g}$ (by $\Phi_{\varepsilon}$) \cite{IW}.

\smallskip

Suppose that for the semidirect product $\frak{g}=\frak{s}%
\overrightarrow{\oplus}_{R}\frak{r}$ \ we can find a virtual copy
of $\frak{s}$ in the enveloping algebra $\mathcal{U}\left(
\frak{g}\right)  $ generated by the operators
\[
X_{i}^{\prime}=X_{i}f\left(  Y_{1},..,Y_{m}\right)  +P_{i}\left(
Y_{1},..,Y_{m}\right)
\]
satisfying the requirements of (\ref{OP1}). Now let
$\frak{g}\rightsquigarrow
\frak{g}^{\prime}=\frak{s}\overrightarrow{\oplus}_{R}\frak{r}^{\prime}$
be a nontrivial generalized In\"{o}n\"{u}-Wigner contraction such
that $\mathcal{N}\left(  \frak{g}\right)  =\mathcal{N}\left(
\frak{g}^{\prime }\right)  $, given by transformations of the type
\cite{We}:
\begin{eqnarray}
\Phi_{\varepsilon}(X_{i})=X_{i},\nonumber\\
\Phi_{\varepsilon}(Y_{j})=\varepsilon^{n_{i}}Y_{j},\quad n_{i}%
\in\mathbb{Z},\label{Kont1}
\end{eqnarray}
where $\left\{  X_{1},..,X_{n},Y_{1},..Y_{m}\right\}  $ is a basis
of $\frak{g}$. This means that the contraction is performed only
in the radical, preserving the Levi part $\frak{s}$ and the
representation $R$. Rewriting the polynomials $f$ and $P_{i}$ over
the transformed basis, we obtain the expressions
\begin{eqnarray}
f(\Phi_{\varepsilon}(Y_{1}),..,\Phi_{\varepsilon}(Y_{m}))
=\varepsilon
^{-\left(  n_{i_{1}}+...+n_{i_{k-1}}\right)  }\alpha^{i_{1}...i_{k-1}}%
Y_{i_{1}}...Y_{i_{k-1}},\nonumber\\
P_{i}(\Phi_{\varepsilon}(Y_{1}),..,\Phi_{\varepsilon}(Y_{m}))
=\varepsilon^{-\left(  n_{j_{1}}+...+n_{j_{k-1}}\right)  }\beta_{i}%
^{j_{1}...j_{k}}Y_{j_{1}}...Y_{j_{k}}.\label{KontOp}
\end{eqnarray}
Considering the maximal exponents of the contraction parameter
$\varepsilon$ defined by
\begin{eqnarray}
M_{0}   =\max\left\{  n_{i_{1}}+...+n_{i_{k-1}}\quad|\quad\alpha
^{i_{1}..i_{k-1}}\neq0\right\}  ,\nonumber\\
M_{i}   =\max\left\{  n_{i_{1}}+...+n_{i_{k}}\quad|\quad\beta_{i}%
^{i_{1}..i_{k-1}}\neq0\right\}  ,
\end{eqnarray}
we can compute the following limits for $i=1,..,n$:
\begin{eqnarray}
\fl f_{0}(Y_{1},..,Y_{m})  &
=\lim_{\varepsilon\rightarrow0}\varepsilon
^{M_{0}}f(\Phi_{\varepsilon}(Y_{1}),..,\Phi_{\varepsilon}(Y_{m}))=\sum
_{n_{i_{1}}+...+n_{i_{p}}=M_{0}}\alpha^{i_{1}...i_{k-1}}Y_{i_{1}}%
...Y_{i_{k-1}},\nonumber\\
\fl P_{i,0}(Y_{1},..,Y_{m})  &
=\lim_{\varepsilon\rightarrow0}\varepsilon
^{M_{i}}P_{i}(\Phi_{\varepsilon}(Y_{1}),..,\Phi_{\varepsilon}(Y_{m}%
))=\sum_{n_{i_{1}}+...+n_{i_{p}}=M_{i}}\beta_{i}^{j_{1}...j_{k}}Y_{j_{1}%
}...Y_{j_{k}}.\label{Lim}
\end{eqnarray}
We further introduce the scalars
\[
N_{i}=\max\left\{  M_{0},M_{i}\right\},\quad i=1,..,n.
\]
This enables us to compute the following limits for arbitrary
index $i$
\begin{equation}
X_{i}^{\prime\prime}=\lim_{\varepsilon\rightarrow0}\varepsilon^{N_{i}}\left\{
X_{i}f(\Phi_{\varepsilon}(Y_{1}),..,\Phi_{\varepsilon}(Y_{m}))+P_{i}%
(\Phi_{\varepsilon}(Y_{1}),..,\Phi_{\varepsilon}(Y_{m}))\right\} .\label{GW1}%
\end{equation}
Observe that if $M_{0}>M_{i_{0}}$ for some $i_{0}$, then the
previous limit reduces to
\[
X_{i_{0}}^{\prime\prime}=X_{i}f_{0}(Y_{1},..,Y_{m}),
\]
while for $M_{0}<M_{i_{0}}$, equation (\ref{GW1}) simplifies to
\[
X_{i_{0}}^{\prime\prime}=P_{i_{0},0}(Y_{1},..,Y_{m}).
\]
In any case, it is obvious that $f_{0}$ is a Casimir operator of
the contraction $\frak{g}^{\prime}$, depending only on the
variables of the radical $\frak{r}^{\prime}$.

\begin{proposition}
A generalized In\"on\"u-Wigner contraction
$\frak{g}\rightsquigarrow \frak{g}^{\prime}$ of type (\ref{Kont1})
determines a virtual copy of $\frak{s}$ in the enveloping algebra
of $\frak{g}^{\prime}$ only if
\[
M_{0}=M_{1}=..=M_{n}.
\]
\end{proposition}

The proof is essentially an adaptation of well known facts on
contractions of Casimir operators \cite{Fra}. Suppose that for
some index $i_{0}\in\left\{  1,..,n\right\}  $ we have
$N_{i_{0}}=M_{0}$. As commented above, the contracted operator is $X_{i_{0}%
}^{\prime\prime}=X_{i_{0}}f_{0}\neq0$. Since the representation
$R$ of $\frak{s}$ is preserved by the contraction, there exists at
least an element $Y_{k_{0}}\in\frak{r}^{\prime}$ such that $\left[
X_{i_{0}},Y_{k_{0}}\right] \neq0$. But this implies that
\begin{equation*}
\left[  X_{i_{0}}^{\prime\prime},Y_{k_{0}}\right]  =\left[  X_{i_{0}}%
f_{0},Y_{k_{0}}\right]  =\left[  X_{i_{0}},Y_{k_{0}}\right]
f_{0}\neq0,
\end{equation*}
contradicting condition (\ref{Bed1}). Thus the $\left\{
X_{i}^{\prime\prime}\right\}  $ cannot generate a virtual copy of $\frak{s}$ in $\mathcal{U}%
\left(  \frak{g}\right)  $. If on the contrary,
$N_{i_{0}}=M_{i_{0}}$, then
the contraction of the operator $X_{i_{0}}^{\prime}$  equals $X_{i_{0}%
}^{\prime\prime}=P_{i_{0},0}$, which depends only on the variables
of
$\frak{r}^{\prime}$. Now $\frak{s}$ is semisimple, thus we can find $X_{k_{1}%
},X_{k_{2}}\in\frak{s}$ such that $\left[
X_{k_{1}},X_{k_{2}}\right] =X_{i_{0}}$. In this case we would
obtain
\begin{equation*}
\left[  X_{k_{1}}^{\prime\prime},X_{k_{2}}^{\prime\prime}\right]  =X_{i_{0}%
}f_{0}^{2}+...\neq X_{i_{0}}^{\prime\prime},
\end{equation*}
which again prevents the operators $\left\{
X_{i}^{\prime\prime}\right\}  $ to generate a virtual copy of
$\frak{s}$. Observe that this is precisely what happens if the
contraction is an inhomogeneous algebra, i.e., $\frak{r}^{\prime}$
is an Abelian algebra.

Let us therefore suppose that all indices $N_{i}$ coincide, i.e.,
that $M_{0}=M_{i}$ for all $i$. Let $f_{1}=\sum_{n_{i_{1}}+...+n_{i_{p}}<M_{0}}\alpha^{i_{1}...i_{k-1}}Y_{i_{1}%
}...Y_{i_{k-1}}$ and $P_{i,1}=\sum_{n_{i_{1}}+...+n_{i_{p}}<M_{i}}\beta_{i}%
^{j_{1}...j_{k}}Y_{j_{1}}...Y_{j_{k}}$. For any $i$ we can
decompose $\varepsilon^{M_{0}}X_{i}^{\prime}$ with respect to the
parameter $\varepsilon$ as
\begin{equation}
\begin{array}
[c]{rl}%
\varepsilon^{M_{0}}X_{i}^{\prime}= & \varepsilon^{M_{0}}\left\{  X_{i}%
f(\Phi_{\varepsilon}(Y_{1}),..,\Phi_{\varepsilon}(Y_{m}))+P_{i}(\Phi
_{\varepsilon}(Y_{1}),..,\Phi_{\varepsilon}(Y_{m}))\right\}  \\
= & X_{i}\left(  f_{0}+\sum_{p\geq1}\varepsilon^{M_{0}-p}f_{1}\right)  +\left( P_{i,0}%
+\sum_{q\geq 1}\varepsilon^{M_{0}-q}P_{i,1}\right).
\end{array}
\end{equation}
Now, over the transformed basis the condition (\ref{Bed1}) must
hold for any $\varepsilon\neq0$. Developing the commutator $\left[
\varepsilon^{M_{0}}X_{i}^{\prime},Y_{k}\right]$ and reordering the
result with respect to $\varepsilon$ leads to to the expression
\begin{equation*}
\fl
\begin{array}
[c]{rl}%
\left[  \varepsilon^{M_{0}}X_{i}^{\prime},Y_{k}\right]  = & \left[
X_{i}\left(  f_{0}+\sum_{p\geq1}\varepsilon^{M_{0}-p}f_{1}\right)  +\left( P_{i,0}%
+\sum_{q\geq 1}\varepsilon^{M_{0}-q}P_{i,1}\right)  ,Y_{k}\right]  \\
= & \left[  X_{i}^{\prime\prime},Y_{k}\right]
+\sum_{p\geq1}\varepsilon^{M_{0}-p}\left[
X_{i}\,f_{1},Y_{k}\right]
+\sum_{q\geq1}\varepsilon^{M_{0}-q}\left[ P_{i,1},Y_{k}\right] =0.
\end{array}
\end{equation*}
Therefore, for the limit $\varepsilon\rightarrow 0$ we finally get
\[
\lim_{\varepsilon\rightarrow0}\left[  \varepsilon^{M_{0}}X_{i}^{\prime}%
,Y_{k}\right]  =\left[  X_{i}^{\prime\prime},Y_{k}\right]  =0,
\]
which shows that the contracted operators  $\left\{
X_{i}^{\prime\prime }\right\}  $ satisfy constraint (\ref{Bed1}).
Arguing in similar manner, it is shown that constraint
(\ref{Bed2}) is also satisfied. We conclude that the $\left\{
X_{i}^{\prime\prime }\right\}  $ generate a virtual copy of
$\frak{s}$ in the enveloping algebra of the contraction
$\frak{g}^{\prime}$. In the particular case of
$\mathcal{N}(\frak{g})=\mathcal{N}(\frak{g}^{\prime})=\mathcal{N}(\frak{s})+1$,
a fundamental set of invariants is completely determined by the
corresponding virtual copies and the Casimir operators $f$ and
$f_{0}$, respectively.

\bigskip
As an example to illustrate the contraction of virtual copies, let
us consider four boson pairs $\left\{
a_{i},a_{i}^{\dagger}\right\}  $ with well-known relations%
\[
\left[  a_{i},a_{j}^{\dagger}\right]  =\delta_{ij},\;\left[  a_{i}%
,a_{j}\right]  =\left[  a_{i}^{\dagger},a_{j}^{\dagger}\right]
=0.
\]
Let $\frak{g}$ be the Lie algebra generated by the operators
\begin{equation*}
\fl
\begin{array}
[c]{llll}%
X_{1,1}=a_{1}^{\dagger}a_{1}+a_{2}^{\dagger}a_{2}, & X_{-1,1}=a_{1}a_{1}%
+a_{2}a_{2}, & X_{1,-1}=a_{1}^{\dagger}a_{1}^{\dagger}+a_{2}^{\dagger}%
a_{2}^{\dagger}, & G_{1}=a_{1}a_{4}+a_{2}a_{3},\\
F_{1}=a_{2}^{\dagger}a_{3}+a_{1}^{\dagger}a_{4}, & Q_{1}=a_{2}a_{4}+a_{1}%
a_{3}, & P_{1}=a_{1}^{\dagger}a_{3}+a_{2}^{\dagger}a_{4} & R=a_{4}a_{4}%
+a_{3}a_{3},\\
E=a_{3}^{\dagger}a_{4}+a_{4}^{\dagger}a_{3}, & T=2a_{3}a_{4}. &  &
\end{array}
\end{equation*}
It is straightforward to verify that the subalgebra generated by
$\left\{ X_{1,1},X_{-1,1},X_{1,-1},F_{1},G_{1},R\right\}  $ is
isomorphic to the semidirect sum $\frak{wsp}\left(
1,\mathbb{R}\right)  $, which incidentally is isomorphic to the
2-photon algebra \cite{Que,Zha}. We can therefore consider
$\frak{g}$ as a kind of ``inhomogenization'' of $\frak{wsp}\left(
1,\mathbb{R}\right)  $, in analogy with the Hamilton algebras.
$\frak{g}$ clearly satisfies $\mathcal{N}\left(  \frak{g}\right)
=2$, and although the centre is zero,\footnote{The commutators of
$\frak{g}$ and the contraction $\frak{g}^{\prime}$ are given in
Appendix B.} the polynomial $R^{2}-T^{2}$ is a Casimir operator of
$\frak{g}$. It can be easily verified that a virtual copy of the
Levi subalgebra $\frak{su}\left(  1,1\right)  $ in the enveloping
algebra of $\frak{g}$ is generated by the operators
\begin{eqnarray*}
X_{1,1}^{\prime}   &=X_{1,1}\left(  R^{2}-T^{2}\right)  +T\left(  Q_{1}%
F_{1}+G_{1}P_{1}\right)  -R\left(  G_{1}F_{1}+Q_{1}P_{1}\right),  \\
X_{-1,1}^{\prime}   &=X_{-1,1}\left(  R^{2}-T^{2}\right)  +2TG_{1}Q_{1}%
-RG_{1}^{2}-RQ_{1}^{2},\\
X_{1,-1}^{\prime}   &=X_{1,-1}\left(  R^{2}-T^{2}\right)  +2TF_{1}P_{1}%
-RF_{1}^{2}-RP_{1}^{2}.
\end{eqnarray*}
In particular, the second Casimir operator of $\frak{g}$ is easily
obtained as
\begin{equation*}
C=\left(  X_{1,1}^{\prime}\right)  ^{2}-\frac{1}{2}\left(
X_{-1,1}^{\prime
}X_{1,-1}^{\prime}+X_{1,-1}^{\prime}X_{-1,1}^{\prime}\right)
\end{equation*}
If we now consider the contraction
$\frak{g}\rightsquigarrow\frak{g}^{\prime}$ determined by the
transformations
\begin{equation*}
\Phi_{\varepsilon}\left(  Q_{1}\right)  =\varepsilon Q_{1},\,\Phi
_{\varepsilon}\left(  P_{1}\right)  =\varepsilon
P_{1},\;\Phi_{\varepsilon }\left(  E\right)  =\varepsilon
E,\,\Phi_{\varepsilon}\left(  T\right) =\varepsilon T,
\end{equation*}
where the remaining generators remain unchanged, we obtain an
algebra $\frak{g}^{\prime}$ with two invariants that also extends
the 2-photon algebra. Expressing the operators generating the
virtual copy in $\mathcal{U}\left(  \frak{g}\right)  $ over the
transformed basis, and considering the previous limits
(\ref{GW1}), we obtain
\begin{eqnarray*}
\varepsilon^{2}X_{1,1}^{\prime}   &=X_{1,1}\left(
\varepsilon^{2}R^{2}-T^{2}\right)
+T\left(  Q_{1}F_{1}+G_{1}P_{1}\right)  -R\left(  \varepsilon^{2}G_{1}%
F_{1}+Q_{1}P_{1}\right),  \\
\varepsilon^{2}X_{-1,1}^{\prime}   &=X_{-1,1}\left(
\varepsilon^{2}R^{2}-T^{2}\right)
+2TG_{1}Q_{1}-R\varepsilon^{2}G_{1}^{2}-RQ_{1}^{2},\\
\varepsilon^{2}X_{1,-1}^{\prime}  & =X_{1,-1}\left(
\varepsilon^{2}R^{2}-T^{2}\right)
+2TF_{1}P_{1}-R\varepsilon^{2}F_{1}^{2}-RP_{1}^{2}.
\end{eqnarray*}
The contracted operators $X_{i,j}^{\prime\prime}$ are therefore
\begin{eqnarray*}
X_{1,1}^{\prime\prime}  & =X_{1,1}\left(  -T^{2}\right)  +T\left(  Q_{1}F_{1}%
+G_{1}P_{1}\right)  -RQ_{1}P_{1},\\
X_{-1,1}^{\prime\prime} &  =X_{-1,1}\left(  -T^{2}\right)  +2TG_{1}Q_{1}-RQ_{1}%
^{2},\\
X_{1,-1}^{\prime\prime} &  =X_{1,-1}\left(  -T^{2}\right)  +2TF_{1}P_{1}%
-RP_{1}^{2}.
\end{eqnarray*}
It is straightforward to verify that these operators generate a
copy of $\frak{su}(1,1)$ in the enveloping algebra of the
contraction $\frak{g}^{\prime}$. In particular, the second Casimir
operator follows from the expression $C^{\prime}=\left(
X_{1,1}^{\prime\prime}\right)  ^{2}-\frac{1}{2}\left(
X_{-1,1}^{\prime\prime}X_{1,-1}^{\prime\prime}+X_{1,-1}^{\prime\prime}X_{-1,1}^{\prime\prime}\right)$,
and coincides with the contraction of the Casimir operator $C$ of
$\frak{g}$.

\section{Final remarks}

We have seen that modifications of the procedure to compute the
Casimir operators of semidirect products
$\frak{s}\overrightarrow{\oplus}_{R}\frak{w}(n)$ using quadratic
operators in the enveloping algebra can be extended naturally to
more general types of Lie algebras, including semidirect products
with pure solvable radicals. This enables us to derive the Casimir
operators of such Lie algebras using the classical formulae for
the invariants of semisimple algebras. The conditions (\ref{Bed1})
and (\ref{Bed2}) also determine the range of validity of the
ansatz based on operators of the form (\ref{OP1}). First of all,
the semidirect product
$\frak{g}=\frak{s}\overrightarrow{\oplus}_{R}\frak{r}$ must
possess at least $\mathcal{N}(\frak{s})+1$ Casimir operators, one
of them depending only on the generators of $\frak{r}$ and
$\mathcal{N}(\frak{s})$ depending on all generators of $\frak{g}$.
Further, the radical $\frak{r}$ of $\frak{g}$ cannot be Abelian.
This implies that the procedure is not directly adaptable to
classical inhomogeneous algebras, at least with the proposed type
of operators. In spite of this apparent restrictions, the class of
Lie algebras to which the method of virtual copies can be applied
is ample. It covers the central extensions of semidirect products
of the type $\frak{s}\overrightarrow{\oplus}_{R}\frak{w}(n)$
(whenever they exist), as well as other ``inhomogenizations"
containing the latter as subalgebras. In particular, it comprises
the one dimensional non-central extensions of double inhomogeneous
Lie algebras \cite{Her}.

\smallskip

We have also shown that for generalized In\"on\"u-Wigner
contractions $\frak{g}\rightsquigarrow
\frak{g}^{\prime}=\frak{s}\overrightarrow{\oplus}_{R}\frak{r}^{\prime}$
that preserve the Levi subalgebra $\frak{s}$ and the
representation $R$, the procedure can be applied to construct a
virtual copy of $\frak{s}$ in the enveloping algebra of
$\frak{g}^{\prime}$. In some sense, this result can be interpreted
as a kind of contraction of realizations of semisimple Lie
algebras. In particular, from this we easily deduce the Casimir
operators of the contraction in compact form.  This leads to the
question whether the ansatz can be reversed, i.e., under which
conditions the realization of a simple Lie algebra $\frak{s}$ in
the enveloping algebra of a semidirect product $\frak{g}^{\prime}$
can be deformed into a realization in another enveloping algebra
$\mathcal{U}(\frak{g})$, and whether the underlying Lie algebras
$\frak{g}$ and $\frak{g}^{\prime}$ are related by contraction
\cite{Fra}. This problem is deeply related with the stability
problem of semidirect products of Lie algebras, a question that
remains largely unsolved.

\section*{Acknowledgment}
During the preparation of this work, the first author (RCS) was
financially supported by the research project MTM2006-09152 of the
M.E.C. and the project and CCG07-UCM/ESP-2922 of the U.C.M.-C.A.M.

\newpage

\section*{Appendix A}

In this appendix we specify the commutators of the basis elements
of $\frak{g}$ with higher order polynomials in the generators. The
relations are presented simultaneously for the inhomogeneous
Hamilton algebra $IHa\left( N\right)  $ and its various central
extensions. To obtain the corresponding commutators for e.g.
$IHa\left(  N\right)  $, all terms containing the generators
$M,A,L$ in the following list are skipped. Analogous procedure
holds for the remaining extensions.

\begin{equation*}
\fl
\begin{tabular*}{\textwidth}{@{}lll*{16}{@{\extracolsep{0pt plus
11pt}}l}}
\multicolumn{2}{l}{$\left[  J_{ij},G_{k}Q_{l}\right]  =\delta_{j}^{k}%
G_{i}Q_{l}-\delta_{i}^{k}G_{j}Q_{l}+\delta_{j}^{l}G_{k}Q_{i}-\delta_{l}%
^{i}G_{k}Q_{j}$,} & $\left[  G_{i}Q_{j},G_{k}\right]  =-\delta_{j}^{k}G_{i}%
T,$\\
\multicolumn{2}{l}{$\left[  J_{ij},F_{k}P_{l}\right]  =\delta_{j}^{k}%
F_{i}P_{l}-\delta_{i}^{k}F_{j}P_{l}+\delta_{j}^{l}F_{k}P_{i}-\delta_{l}%
^{i}F_{k}P_{j},$} & $\left[  G_{i}Q_{j},F_{k}\right]  =-\delta_{j}^{k}%
G_{i}A+\delta_{i}^{k}Q_{j}R,$\\
\multicolumn{2}{l}{$\left[  J_{ij},G_{k}F_{l}\right]  =\delta_{j}^{k}%
G_{i}F_{l}-\delta_{i}^{k}G_{j}F_{l}+\delta_{j}^{l}G_{k}F_{i}-\delta_{l}%
^{i}G_{k}F_{j}$,} & $\left[  G_{i}Q_{j},Q_{k}\right]  =\delta_{k}^{i}Q_{j}%
T,$\\
\multicolumn{2}{l}{$\left[  J_{ij},P_{k}Q_{l}\right]  =\delta_{j}^{k}%
P_{i}Q_{l}-\delta_{i}^{k}P_{j}Q_{l}+\delta_{j}^{l}P_{k}Q_{i}-\delta_{l}%
^{i}P_{k}Q_{j}$,} & $\left[  G_{i}Q_{j},P_{k}\right]  =\delta_{i}^{k}%
Q_{j}M-\delta_{k}^{j}G_{i}L,$\\
\multicolumn{2}{l}{$\left[  J_{ij},Q_{k}F_{l}\right]  =\delta_{j}^{k}%
Q_{i}F_{l}-\delta_{i}^{k}Q_{j}F_{l}+\delta_{j}^{l}Q_{k}F_{i}-\delta_{l}%
^{i}Q_{k}F_{j}$,} & $\left[  G_{i}Q_{j},E\right]  =P_{i}Q_{j},$\\
\multicolumn{2}{l}{$\left[  J_{ij},P_{k}G_{l}\right]  =\delta_{j}^{k}%
P_{i}G_{l}-\delta_{i}^{k}P_{j}G_{l}+\delta_{j}^{l}P_{k}G_{i}-\delta_{l}%
^{i}P_{k}G_{j}$,} & $\left[  F_{i}P_{j},G_{k}\right]  =-\delta_{j}^{k}%
F_{i}M-\delta_{i}^{k}P_{j}R,$\\
\multicolumn{2}{l}{$\left[  G_{i}Q_{j},G_{k}Q_{l}\right]  =\delta_{i}%
^{l}TG_{k}Q_{j}-\delta_{k}^{j}TG_{i}Q_{l}$,} & $\left[  F_{i}P_{j}%
,F_{k}\right]  =-\delta_{k}^{j}F_{i}T,$\\
\multicolumn{2}{l}{$\left[  G_{i}Q_{j},F_{k}P_{l}\right]  =\delta_{i}%
^{k}RQ_{j}P_{l}-\delta_{l}^{j}LG_{i}F_{k}+\delta_{i}^{l}MF_{k}Q_{j}+\delta
_{k}^{i}\delta_{j}^{l}RL$,} & $\left[  F_{i}P_{j},Q_{k}\right]
=\delta
_{i}^{k}P_{j}A+\delta_{k}^{j}F_{i}L,$\\
\multicolumn{2}{l}{$\left[  G_{i}Q_{j},G_{k}F_{l}\right]  =\delta_{i}%
^{l}RG_{k}Q_{j}-\delta_{j}^{k}TG_{i}F_{l}-\delta_{l}^{j}AG_{i}G_{k}$,}
&
$\left[  F_{i}P_{j},P_{k}\right]  =\delta_{i}^{k}P_{j}T,$\\
\multicolumn{2}{l}{$\left[  G_{i}Q_{j},P_{k}Q_{l}\right]  =\delta_{i}%
^{l}TP_{k}Q_{j}-\delta_{j}^{k}LG_{i}Q_{l}+\delta_{i}^{k}MQ_{j}Q_{l}$,}
&
$\left[  F_{i}P_{j},E\right]  =-Q_{i}P_{j},$\\
\multicolumn{2}{l}{$\left[  G_{i}Q_{j},Q_{k}F_{l}\right]  =\delta_{i}%
^{k}TQ_{j}F_{l}-\delta_{j}^{l}AG_{i}Q_{k}+\delta_{i}^{k}\delta_{j}^{l}TA$,}
&
$\left[  G_{i}F_{j},G_{k}\right]  =-\delta_{k}^{j}G_{i}R,$\\
\multicolumn{2}{l}{$\left[  G_{i}Q_{j},P_{k}G_{l}\right]  =\delta_{i}%
^{k}MQ_{j}G_{l}-\delta_{j}^{k}LG_{i}G_{l}-\delta_{l}^{j}TG_{i}P_{k}+\delta
_{k}^{i}\delta_{l}^{j}MT$,} & $\left[  G_{i}F_{j},F_{k}\right]
=\delta
_{i}^{k}F_{j}R,$\\
\multicolumn{2}{l}{$\left[  F_{i}P_{j},F_{k}P_{l}\right]  =\delta_{i}%
^{l}TF_{k}P_{j}-\delta_{k}^{j}TF_{i}P_{l}$,} & $\left[  G_{i}F_{j}%
,Q_{k}\right]  =\delta_{i}^{k}F_{j}T+\delta_{k}^{j}G_{i}A,$\\
\multicolumn{2}{l}{$\left[  F_{i}P_{j},G_{k}F_{l}\right]  =-\delta_{j}%
^{l}TF_{i}G_{j}-\delta_{j}^{k}MF_{i}F_{l}-\delta_{k}^{j}RP_{j}F_{k}-\delta
_{k}^{i}\delta_{l}^{j}RT$,} & $\left[  G_{i}F_{j},P_{k}\right]
=\delta
_{j}^{k}G_{i}T+\delta_{k}^{i}F_{j}M,$\\
\multicolumn{2}{l}{$\left[  F_{i}P_{j},P_{k}Q_{l}\right]  =\delta_{j}%
^{l}LF_{i}P_{k}+\delta_{i}^{k}TQ_{l}P_{j}+\delta_{i}^{l}AP_{j}P_{k}$,}
&
$\left[  G_{i}F_{j},E\right]  =P_{i}F_{j}-G_{i}Q_{j},$\\
\multicolumn{2}{l}{$\left[  F_{i}P_{j},Q_{k}F_{l}\right]  =\delta_{k}%
^{j}LF_{i}F_{l}-\delta_{l}^{j}TF_{i}Q_{k}+\delta_{k}^{i}AF_{l}P_{j}$,}
&
$\left[  P_{i}Q_{j},G_{k}\right]  =-\delta_{k}^{j}P_{i}T-\delta_{k}^{i}%
Q_{j}M,$\\
\multicolumn{2}{l}{$\left[  F_{i}P_{j},Q_{k}F_{l}\right]  =\delta_{k}%
^{j}LF_{i}F_{l}-\delta_{l}^{j}TF_{i}Q_{k}+\delta_{i}^{k}AF_{l}P_{j}$,}
&
$\left[  P_{i}Q_{j},F_{k}\right]  =-\delta_{j}^{k}P_{i}A-\delta_{i}^{k}%
Q_{j}T,$\\
\multicolumn{2}{l}{$\left[  F_{i}P_{j},P_{k}G_{l}\right]  =\delta_{k}%
^{i}TP_{j}G_{l}-\delta_{l}^{j}MP_{k}F_{i}-\delta_{l}^{i}RP_{k}P_{j}$,}
&
$\left[  P_{i}Q_{j},Q_{k}\right]  =\delta_{i}^{k}Q_{j}L,$\\
\multicolumn{2}{l}{$\left[  G_{i}F_{j},G_{k}F_{l}\right]  =\delta_{l}%
^{i}RG_{k}F_{j}-\delta_{j}^{k}RG_{i}F_{l}$,} & $\left[  P_{i}Q_{j}%
,P_{k}\right]  =-\delta_{j}^{k}P_{i}L,$\\
\multicolumn{2}{l}{$\left[  G_{i}F_{j},P_{k}Q_{l}\right]  =\delta_{k}%
^{j}TG_{i}Q_{l}+\delta_{j}^{l}AG_{i}P_{k}+\delta_{k}^{i}MQ_{l}F_{j}+\delta
_{i}^{l}TP_{k}F_{j}$,} & $\left[  Q_{i}F_{j},G_{k}\right]  =-\delta_{j}%
^{k}Q_{i}R-\delta_{k}^{i}F_{j}T,$\\
\multicolumn{2}{l}{$\left[  G_{i}F_{j},Q_{k}F_{l}\right]  =\delta_{k}%
^{j}AG_{i}F_{l}+\delta_{k}^{i}TF_{j}F_{l}+\delta_{i}^{l}RQ_{k}F_{j}$,}
&
$\left[  Q_{i}F_{j},F_{k}\right]  =-\delta_{k}^{i}F_{j}A,$\\
\multicolumn{2}{l}{$\left[  G_{i}F_{j},P_{k}G_{l}\right]  =\delta_{k}%
^{j}TG_{i}G_{l}-\delta_{l}^{j}RG_{i}P_{k}+\delta_{k}^{i}MF_{j}G_{l}+\delta
_{k}^{i}\delta_{j}^{l}MR$,} & $\left[  Q_{i}F_{j},Q_{k}\right]
=\delta
_{k}^{j}Q_{i}A,$\\
\multicolumn{2}{l}{$\left[  P_{i}Q_{j},P_{k}Q_{l}\right]  =\delta_{l}%
^{i}LP_{k}Q_{j}-\delta_{k}^{j}LP_{i}Q_{l}$,} & $\left[  Q_{i}F_{j}%
,P_{k}\right]  =\delta_{j}^{k}Q_{i}T-\delta_{k}^{i}F_{j}L,$\\
\multicolumn{2}{l}{$\left[  P_{i}Q_{j},Q_{k}F_{l}\right]  =\delta_{k}%
^{i}LF_{l}Q_{j}-\delta_{l}^{j}AP_{i}Q_{k}-\delta_{l}^{i}TQ_{k}Q_{j}$,}
&
$\left[  Q_{i}F_{j},E\right]  =-Q_{i}Q_{j},$\\
\multicolumn{2}{l}{$\left[  P_{i}Q_{j},P_{k}G_{l}\right]  =-\delta_{k}%
^{j}LP_{i}G_{l}-\delta_{l}^{j}TP_{i}P_{k}-\delta_{i}^{l}MP_{k}Q_{j}$,}
&
$\left[  P_{i}G_{j},G_{k}\right]  =-\delta_{k}^{i}G_{j}M,$\\
\multicolumn{2}{l}{$\left[  Q_{i}F_{j},Q_{k}F_{l}\right]  =\delta_{k}%
^{j}AQ_{i}F_{l}-\delta_{l}^{i}AQ_{k}F_{v}$,} & $\left[  P_{i}G_{j}%
,F_{k}\right]  =\delta_{k}^{i}P_{i}R-\delta_{k}^{i}G_{j}T,$\\
\multicolumn{2}{l}{$\left[  Q_{i}F_{j},P_{k}Q_{l}\right]  =\delta_{k}%
^{j}TQ_{i}G_{l}-\delta_{i}^{k}LF_{j}G_{l}-\delta_{l}^{j}RP_{k}Q_{i}-\delta
_{i}^{l}TP_{k}F_{j}$,} & $\left[  P_{i}G_{j},Q_{k}\right]  =\delta_{k}%
^{j}P_{i}T+\delta_{i}^{k}G_{j}L,$\\
\multicolumn{2}{l}{$\left[  P_{i}G_{j},P_{k}G_{l}\right]  =\delta_{k}%
^{j}MP_{i}G_{l}-\delta_{i}^{l}MP_{k}G_{j}$,} & $\left[  P_{i}G_{j}%
,P_{k}\right]  =\delta_{j}^{k}P_{i}M,$\\
\multicolumn{2}{l}{$\left[  P_{i}G_{j},E\right]  =P_{i}P_{j},\quad
\left[  T^{2}+RL,E\right]  =0,$}
& $\left[  TF_{i}P_{j},E\right]  =TQ_{i}P_{j}+LQ_{i}G_{j},$\\
\multicolumn{2}{l}{$\left[  TG_{i}Q_{j},E\right]
=TP_{i}Q_{j}+LG_{i}Q_{j},\quad \left[ RP_{i}Q_{j},E\right]
=-2TP_{i}Q_{j}$.} & \\
\end{tabular*}
\end{equation*}

\section*{Appendix B}

The Lie algebra $\frak{g}$ spanned by the operators $\left\{  X_{1,1}%
,X_{-1,1},X_{1,-1},F_{1},G_{1},Q_{1},P_{1},R,E,T\right\}  $ has
Levi decomposition
\begin{center}
$\frak{g}=\frak{su}\left(  1,1\right)  \overrightarrow{\oplus}_{R}%
\frak{r}$,
\end{center}
where $R=D_{\frac{1}{2}}\oplus D_{\frac{1}{2}}\oplus D_{0}^{3}$,
$D_{\frac{1}{2}}$ being the two dimensional irreducible spin $\frac{1}{2}%
$-representation and $D_{0}$ the trivial representation. The
commutators of $\frak{g}$ and the contraction $\frak{g}^{\prime}$
are simultaneously given in the following table, for the values of
$\alpha=1$ and $\alpha=0$, respectively.

\smallskip

\begin{center}
$\begin{array}
[c]{c|cccccccccc}%
\left[\;,\;\right]& X_{1,1} & X_{-1,1} & X_{1,-1} & G_{1} & F_{1}
& Q_{1} & P_{1} & R & E & T\\\hline X_{1,1} & 0 & -2X_{-1,1} &
2X_{1,-1} & -G_{1} & F_{1} & -Q_{1} & P_{1} & 0 &
0 & 0\\
X_{-1,1} &  & 0 & 4X_{1,1} & 0 & 2G_{1} & 0 & 2Q_{1} & 0 & 0 & 0\\
X_{1,-1} &  &  & 0 & -2F_{1} & 0 & -2P_{1} & 0 & 0 & 0 & 0\\
G_{1} &  &  &  & 0 & R & 0 & T & 0 & Q_{1} & 0\\
F_{1} &  &  &  &  & 0 & -T & 0 & 0 & P_{1} & 0\\
Q_{1} &  &  &  &  &  & 0 & \alpha R & 0 & \alpha G_{1} & 0\\
P_{1} &  &  &  &  &  &  & 0 & 0 & \alpha F_{1} & 0\\
R &  &  &  &  &  &  &  & 0 & 2T & 0\\
E &  &  &  &  &  &  &  &  & 0 & -2\alpha R\\
T &  &  &  &  &  &  &  &  &  & 0
\end{array}$
\end{center}

\end{document}